\documentclass[11pt,a4paper]{article}
\usepackage{graphicx}

\usepackage{amsfonts,amsmath,latexsym,amssymb,amscd,amsthm}
\usepackage[width=15cm]{geometry}

\newtheorem{theorem}{Theorem}

\newtheorem{assumption}{Assumption}
\newtheorem{corollary}{Corollary}
\newtheorem{example}{Example}

\newcommand{\bbeta}{{\boldsymbol{\beta}}}
\newcommand{\btheta}{\boldsymbol{\theta}}
\newcommand{\spaceX}{{\Xi}}
\begin{document}
\title{ KL-optimum designs: theoretical properties\\ and practical computation\thanks{The final publication is available at Springer via http://dx.doi.org/10.1007/s11222-014-9515-8}}
\author{Giacomo Aletti \and Caterina May \and Chiara Tommasi%
}
\date{August 22, 2014}
\maketitle

\begin{abstract}
In this paper some new properties and computational tools for finding KL-optimum designs are provided.
KL-optimality is a  general criterion useful to
select the best experimental conditions to discriminate between statistical models.
A KL-optimum design is obtained from a minimax optimization problem, which is defined on a infinite-dimensional space.
In particular, continuity of the KL-optimality criterion is proved under mild conditions; as a consequence, the first-order algorithm converges to the set of KL-optimum designs for a large class of models.  It is also shown that KL-optimum designs are invariant to any scale-position transformation. Some examples are given and discussed, together with some practical implications for numerical computation purposes.
\end{abstract}

\section{Introduction}
In  presence of families of competing models,
one of the main tasks of the optimum experimental design is the choice of the best experimental conditions to discriminate between rival models.

This problem was studied by many researchers in the field. See for instance, \cite{AtkinsonCox1974}, \cite{Dette1994,Dette} and \cite{Hill1978}, among many others.
In \cite{AF1975a,AF1975b}, the authors provided the T-optimality criterion
to select between two, or more, competing homoscedastic Gaussian models. This criterion was extended in \cite{UB2005} to
heteroscedastic Gaussian  models.
More recently, the KL-optimality, which is based on the Kullback-Leibler divergence, has been proposed in \cite{Tomm2007}.
The idea of using the Kullback-Leibler divergence to construct an optimality criterion is not new. For instance,
 \cite{BoxHill1967} proposed a Bayesian sequential method to discriminate among several  models; at each stage of their sequential scheme, the experimental conditions were chosen by maximizing
a weighted sum of Kullback-Leibler divergences between any couple of models. In this way,
 the maximum change in entropy expected from the observations to be taken at that stage was maximized.
See also \cite{MeyerSteinbergBox1996}, who applied the Box and Hill criterion to the problem of augmenting a multifactor design. Finally, we refer the reader to \cite{Wiens2009} for another optimality criterion based on the Kullback-Leibler divergence.

KL-optimum designs proposed by \cite{Tomm2007} are obtained from a minimax optimization problem. The KL-criterion is very general since it can be applied when the rival models are nested or not, homoscedastic or heteroscedastic and with any error distribution; moreover, it may be seen as an extension of the T-criterion and its generalizations. It was also applied to discriminate among several models in \cite{TommasiModa8} and used in compound criteria for the double goal of discrimination and estimation in \cite{Tomm2009,TommasiMay}.
T- and KL-optimum designs are usually computed numerically. Differently, \cite{DetteMelasShpilev} provide analytical results for T-optimal designs when the interest is in discriminating between two polynomial models which differ in the degree of two.

In this paper, some interesting properties of KL-optimum designs are proved. The notation and framework of the problem are given in Section \ref{sect:frame}.
In Section \ref{sect:invcont} it is shown that KL-optimum designs are invariant to a scale-position transformation. Moreover, it is proved the continuity of the KL-optimality criterion with respect to the experimental design; mild conditions are assumed.
This result is very important also for the numerical computation of a KL-optimum design.
In fact, the continuity property is crucial to prove that the first-order algorithm, recalled in Section \ref{sect:alg}, converges to the set of KL-optimum designs, whenever it moves on regular designs (for a detailed proof see \cite{AMT_Moda10}). In Section \ref{sect:GLM} it is shown that, at least for the case of Generalized Linear Models, any design with non-singular Fisher information matrix is KL-regular and viceversa. As a consequence, the algorithm always converges at least in this class of models. In Section \ref{sect:ex} some examples are provided and discussed. In addition, some computational hints for the application of the algorithm  are described. A final discussion concludes the paper.

\section{Framework and notations}\label{sect:frame}
Let $f_1(y|{\bold x};{\bbeta}_1)$ and $f_{2}(y|{\bold x};{\bbeta}_{2})$ be two rival statistical models, i.e. two parametric families of conditional probability densities of an experimental response $Y$ under the experimental condition ${\bold x}$, where
${\bold x}$ belongs to a compact experimental
domain ${\cal X} \subset \mathbb{R}^q$, $q \ge 1$, the parameters
 ${\bbeta}_i\in\Theta_i$, $i=1,2$
and  $\Theta_i$ is an open set of $\mathbb{R}^{d_i}$, $i=1,2$.\footnote{Note that
the parameter spaces $\Theta_1$ and $\Theta_2$ can be required instead to be compact.}

A design $\xi$ is a probability distribution having support on ${\cal X}$.
In a discrimination problem,
the choice of $\xi$ should be done in order to maximize the ``separation'' between the competing models.

The KL-optimality criterion is based on the Kullback-Leibler divergence between the two conditional distributions
$f_1(y|{\bold x};{\bbeta}_1)$ and $f_{2}(y|{\bold x};{\bbeta}_{2})$:
\begin{equation}
{\cal I}({\bold x},{\bbeta}_1,{\bbeta}_{2})=\int_{\cal Y}  \log
\frac{f_1(y|{\bold x};{\bbeta}_1)}{f_{2}(y|{\bold x};{\bbeta}_{2})}\, f_1(y|{\bold x};{\bbeta}_1)\; dy.
\label{I-quantity2}
\end{equation}
Note that the statistical models $f_1(y|{\bold x};{\bbeta}_1)$ and $f_{2}(y|{\bold x};{\bbeta}_{2})$ may be conditional densities with respect to a common general measure $\mu$; to include discrete models it is enough to replace $dy$  with $\mu(dy)$ in Equation~\eqref{I-quantity2} and in the rest of the paper.
The notation used in \eqref{I-quantity2} is maintained for simplicity. 

The quantity in Equation~\eqref{I-quantity2} is known to be non-negative, and it is zero if
and only if the two responses are equal almost surely. The Kullback-Leibler divergence is
often called distance, although it is not symmetric and does not satisfy the triangular
inequality. In this context,
the Kullback-Leibler divergence in Equation~\eqref{I-quantity2} measures the dissimilarity
between the two different distributions with parameters $\bbeta_1$ and $\bbeta_2$,
when the experimental condition is ${\bold x}$.

If a design $\xi$ is chosen to maximize the power function in the worst case in a hypothesis test where $f_1(y|{\bold x};{\bbeta}_1)$
is the true completely known model  under the alternative $f_2(y|{\bold x};{\bbeta}_2)$, then this design is the maximum
of the KL--optimality criterion proposed
in \cite{Tomm2007}:
\begin{equation}
I_{2,1}(\xi; {\bbeta}_1)= \inf_{{\bbeta}_{2}\in \Theta_{2}}  \int_{\cal X}
{\cal I}({\bold x},{\bbeta}_1,{\bbeta}_{2}) \,\xi(d{\bold x}).
\label{kl1}
\end{equation}
The criterion \eqref{kl1} is the minumum Kullback-Leibler distance between the joint distribution
$f_1(y|{\bold x};{\bbeta}_1)\xi({\bold x})$ and the joint statistical model $f_{2}(y|{\bold x};{\bbeta}_{2})\xi({\bold x})$.

For a given value ${\bbeta}_1 \in \Theta_{1}$ of the first model, a design $\xi$ is {\it KL-regular} if the following set
\begin{equation}\label{eq:regular}
\Omega_{2}(\xi;\bbeta_1)=\Big\{{\hat\bbeta}_{2}:\;{\hat{\bbeta}}_{2}(\xi)=\arg
\inf_{{\bbeta}_{2} \in \Theta_{2}} \int_{\cal X} {\cal I}({\bold x},{\bbeta}_1,{\bbeta}_{2})
\,\xi(d{\bold x})
\Big\}
\end{equation}
is a singleton. Otherwise $\xi$ is called {\it KL-singular}.

A KL-optimum design
\begin{equation}
\xi^* \in \arg\max_\xi I_{2,1}(\xi; {\bbeta}_1)
\label{eq:KL}
\end{equation}
 exists since the KL-criterion function
 \begin{equation}\label{eq:crit}
I_{2,1}(\xi; {\bbeta}_1): ({\spaceX},d_w)\rightarrow [0,+\infty)
\end{equation}
is concave (as proved in \cite{Tomm2007}) and upper semi-continuous (as proved in \cite{TommasiMay}), where
 ${\spaceX}$ is the set of probability distributions
$\xi$ with support ${\cal X }\subset \mathbb{R}^q$ endowed with a metric $d_w$
which metrizes the weak convergence on $\cal X$.
\noindent Since $\cal X$ is compact, the metric space $({\spaceX},d_w)$, which is an infinite-dimensional space, is complete and compact. In fact, any sequence of probability distribution on $\cal X$ is tight, and hence, by Prokhorov's Theorem,
it admits a converging subsequence (in $({\spaceX},d_w)$). In what follows, we take the Kantorovich-Wasserstein metric  (see \cite{Gibbs02}):
\begin{equation*}
d_w(\xi_1,\xi_2)=\inf \{ E(|X_1-X_2|):X_1\sim \xi_1, X_2\sim \xi_2\}.
\end{equation*}
Since ${\bbeta}_1$ is assumed to be known, for sake of simplicy, $I_{2,1}(\xi; {\bbeta}_1)$ will be denoted by $I_{2,1}(\xi)$.

\section{Invariance property and continuity of the KL-criterion}\label{sect:invcont}

In this section two important theoretical properties are presented. First, we prove that  KL-optimum designs are invariant to a
scale-position transformation of the design region, whenever this transformation on the experimental condition
results in a new parametrization of the rival models. Then, we prove the continuity of the KL-optimality criterion \eqref{eq:crit},
with respect to the design $\xi$.\\
\begin{theorem}\label{theo:inv}
Let
$$
{\cal Z}=\left\{ {\bold z}={\bold a}+{\bold B} {\bold x}|{\bold x} \in {\cal X} \right\}
$$
be a rescaled experimental domain, where  ${\bold a}$  and  ${\bold B}$ are, respectively, a  column vector and a non singular matrix of known constants.\\
If
\begin{enumerate}
\item \label{hp:1}
$f_i(y|{\bold x};{\bbeta}_i)=f_i(y|{\bold z};{\bold \gamma}_i)$, where ${\bold \gamma}_i=g_i({\bbeta}_i)$ and $g_i(\cdot)$ is a one to one function,
and $i=1,2$;
\item \label{hp:2}
the following mean log-likelihood
\begin{equation}\label{eq:ml}
{{\cal ML}}_2(\bbeta_2)= \int_{\cal X}\int_{\cal Y} \log f_2(y|{\bold x};{\bbeta}_2) \, f_1(y|{\bold x};{\bbeta}_1)\,dy\,\xi(d{\bold x})
\end{equation}
has at least a maximizer
$\hat{\bbeta}_2 \in \mathop{\arg\sup}\limits_{{\bbeta}_2} {{\cal ML}}_2(\bbeta_2)$ for any fixed design $\xi\in{\spaceX}$;
\end{enumerate}
then a KL-optimum design on  ${\cal Z}$ is
$$
\eta^*=\xi^*\circ {\bold z}^{-1},\quad {\bold x}\in {\cal X}\quad \mbox{and}\quad {\bold z}({\bold x})={\bold a}+{\bold B}  {\bold x}.
$$
\end{theorem}
\begin{proof}
From hypothesis 1 the rival models can be expressed as
$f_i(y|{\bold z};{\bold \gamma}_i)$
where ${\bold z}={\bold a}+{\bold B}  {\bold x}$, $i=1,2$.
The KL-criterion applied to models $f_1(y|{\bold z};{\bold \gamma}_1)$ and $f_2(y|{\bold z};{\bold \gamma}_2)$ is
\begin{eqnarray*}
I_{2,1}^{\cal Z}(\eta)
&=&  \inf_{{\bold \gamma}_2} \int_{\cal Z}\int_{\cal Y}
        \log \frac{f_1(y|{\bold z};{\bold \gamma}_1)}{f_2(y|{\bold z};{\bold \gamma}_2)} \, f_1(y|{\bold z};{\bold \gamma}_1)\,dy\,\eta(d{\bold z})\\
&=&  \int_{\cal Z}\int_{\cal Y} \log f_1(y|{\bold z};{\bold \gamma}_1) \, f_1(y|{\bold z};{\bold \gamma}_1)\,dy\,\eta(d{\bold z}) \\
&-&   \sup_{{\bold \gamma}_2} \int_{\cal Z}\int_{\cal Y} \log f_2(y|{\bold z};{\bold \gamma}_2) \, f_1(y|{\bold z};{\bold \gamma}_1)\,dy\,\eta(d{\bold z}).
\end{eqnarray*}
Let
$$
\widetilde{{\cal ML}}_2({\bold \gamma}_2)= \int_{\cal Z}\int_{\cal Y} \log f_2(y|{\bold z};{\bold \gamma}_2) \, f_1(y|{\bold z};{\bold \gamma}_1)\,dy\,\eta(d{\bold z});
$$
from hypothesis~\ref{hp:1} and setting $\xi(d{\bold x})=\eta(d{\bold z})$
$$
\widetilde{{\cal ML}}_2({\bold \gamma}_2)={{\cal ML}}_2(g_2^{-1}({\bold \gamma}_2))
$$
therefore
\begin{equation}\label{eq:mle}
\hat{\bold \gamma}_2 \in \arg\sup \widetilde{{\cal ML}}_2({\bold \gamma}_2) \qquad \text{if and only if}\qquad
\hat{{\bbeta}}_2 \in \arg\sup {{\cal ML}}_2({\bbeta}_2)
\end{equation}
where $\hat{\bold \gamma}_2 = g_2(\hat{\bbeta}_2)$.

Using hypothesis~\ref{hp:1} and Equation~\eqref{eq:mle},
\begin{eqnarray*}
I_{2,1}^{\cal Z}(\eta)
&=& \int_{\cal Z}\int_{\cal Y} \log f_1(y|{\bold z};{\bold \gamma}_1) \, f_1(y|{\bold z};{\bold \gamma}_1)\,dy\,\eta(d{\bold z}) \\
&-&   \int_{\cal Z}\int_{\cal Y} \log f_2(y|{\bold z};\hat{\bold \gamma}_2) \, f_1(y|{\bold z};{\bold \gamma}_1)\,dy\,\eta(d{\bold z})\\
&=&  \int_{\cal X}\int_{\cal Y} \log f_1(y|{\bold x};{\bbeta}_1) \, f_1(y|{\bold x};{\bbeta}_1)\,dy\,\xi(d{\bold x}) \\
&-&   \int_{\cal X}\int_{\cal Y} \log f_2(y|{\bold x};\hat{{\bbeta}}_2) \, f_1(y|{\bold x};{\bbeta}_1)\,dy\,\xi(d{\bold x})=I_{2,1}^{\cal X}(\xi),
\end{eqnarray*}
which proves the theorem.
\end{proof}


Continuity is a very nice property which is necessary to prove the convergence of the first order algorithm; see \cite{AMT_Moda10}, also recalled in Section \ref{sect:alg}.
We require that the Kullback-Leibler divergence between the two conditional distributions is continuous and also Lipschitz with respect to the experimental condition.
\begin{assumption}\label{ass:contdens2}
The Kullback-Leibler divergence ${\cal I}({\bold x},{\bbeta}_1,{\bbeta}_{2})$ given in Equation~\eqref{I-quantity2}
is a Lipschitz function with respect to ${\bold x}$.
\end{assumption}
Note that the Kullback-Leibler divergence between different models is the difference between the cross entropy of the two models  and the information entropy of the first one. The information entropy is a continuously differentiable function of the parameters in almost all the classes of parametric models.
This smoothness is also observed for the cross entropy between different classes of models with the same support (for instance, Lognormal, Weibull, Gamma, \ldots). Therefore, since $\spaceX$ is compact, if the parameters of the rival models (which are assumed to depend on ${\bold x}$ through some coefficients) are continuously differentiable functions of the  experimental condition ${\bold x}$, then the Kullback-Leibler divergence is a continuously differentiable function of the  experimental conditions, and hence a Lipschitz function.

Denote by ${\cal J}(\xi,\bbeta_1,\bbeta_2)$ the average of the function
${\cal I}({\bold x},\bbeta_1,\bbeta_2)$ in \eqref{I-quantity2} with respect to the probability measure $\xi$, namely
\[
{\cal J}(\xi,\bbeta_1,\bbeta_2) =
\int_{\cal X}
{\cal I}({\bold x},\bbeta_1,\bbeta_2)\, d\xi({\bold x}) =
\int_{\cal X} \int_{\cal Y} \log
\frac{f_1(y|{\bold x};{\bbeta}_1)}{f_{2}(y|{\bold x};{\bbeta}_{2})}\; f_1(y|{\bold x};{\bbeta}_1) \, dy\,\xi(d{\bold x}).
\]
Again, since $\cal X$ is compact, Assumption~\ref{ass:contdens2} implies that
the function ${\cal J}$ is continuous with respect to $\xi$. Moreover, this function is linear in $\xi$.
The problem of finding a KL-optimal design as in equation \eqref{eq:KL} is
an infinite dimension minmax problem.
Our goal
is to prove that
\(
I_{2,1}(\xi) = \inf_{\bbeta_2} {\cal J}(\xi,\bbeta_1,\bbeta_2)
\)
is continuous as an extension of classical results for semi-infinite problem (see for instance
\cite{Polak}) to our context.
We start with a counterexample, which shows that Assumption~\ref{ass:contdens2} is not sufficient for
$I_{2,1}(\xi)$ to be continuous.

\begin{example}[$I_{2,1}(\xi; {\bbeta}_1)$ is not continuous] 
Take ${\cal X} = [0,1]$,
$\Theta_2 = (0,\infty)$, and define
\[
{\cal I}({\bold x},{\bbeta}_1,\bbeta_2) =
\begin{cases}
2((2\bbeta_2 -1)  {\bold x} + (1-\bbeta_2))  & \text{if } 0 < \bbeta_2\leq 1 \\
(\bbeta_2+1) {\bold x}^{\bbeta_2} & \text{if } 1 < \bbeta_2
\end{cases}
\]
We have:
\begin{itemize}
\item
${\cal I}({\bold x},{\bbeta}_1,{\bbeta}_{2})$ is a continuous function
on ${\cal X}\times\Theta_2$;
\item
${\cal I}({\bold x},{\bbeta}_1,{\bbeta}_{2})$ is a convex and Lipschitz function
of ${\bold x}$, for any ${\bbeta}_{2}\in\Theta_2$;
\item $I_{2,1}(\delta_{\bold x}; {\bbeta}_1)=0$ for any ${\bold x}\in{\cal X}$.
\end{itemize}
Take $\xi_n$ be the uniform distribution on $[0, 1-1/n]$; it can be easily proved
that $d_w(\xi_n,\xi)\to 0$, where $\xi$ is the uniform distribution on $[0, 1]$.
We have
\[
\begin{aligned}
\int_{\cal X} {\cal I}({\bold x},{\bbeta}_1,\bbeta_2) d\xi_n({\bold x}) &=
\int_{0}^{1-1/n} \frac{{\cal I}({\bold x},{\bbeta}_1,\bbeta_2)}{1-1/n} d{\bold x} \\
& =
\begin{cases}
1 - \frac{2\bbeta_2 -1}{n}  & \text{if } 0 < \bbeta_2\leq 1 \\
(1-1/n)^{\bbeta_2} & \text{if } 1 < \bbeta_2
\end{cases}
\end{aligned}
\]
while
\[
\int_{\cal X} {\cal I}({\bold x},{\bbeta}_1,\bbeta_2) d\xi({\bold x}) =
\int_{0}^{1} {\cal I}({\bold x},{\bbeta}_1,\bbeta_2) d{\bold x} \equiv 1.
\]
Hence,
\[
\begin{aligned}
I_{2,1}(\xi_n) &= \inf_{\bbeta_2\in\Theta_2}
\int_{\cal X} {\cal I}({\bold x},{\bbeta}_1,\bbeta_2) d\xi_n({\bold x}) = 0\\
& \neq
1 = \inf_{\bbeta_2\in\Theta_2}
\int_{\cal X} {\cal I}({\bold x},{\bbeta}_1,\bbeta_2) d\xi({\bold x}) =
I_{2,1}(\xi).
\end{aligned}
\]
\end{example}

We give here a mild assumption which is satisfied in many situations.
In fact, when we fix $\bbeta_1$, we can expect
that that the Kullback-Leibler divergence ${\cal I}({\bold x},{\bbeta}_1,{\bbeta}_{2})$ is ``dominated'' in $\bbeta_2$, in that
if ${\cal I}({\bold x},{\bbeta}_1,{\bbeta}_{2})$ is too big for some
${\bold x}$, there is
another model $f_2(y|{\bold x};\widetilde{{\bbeta}}_{2})$ that is always
closer to $f_1(y|{\bold x};{\bbeta}_1)$ and dominated by a constant $M({\bbeta}_1)$.

\begin{assumption}\label{ass:contdens3}
For any fixed ${\bbeta}_1$, there exists $M=M(\bbeta_1)>0$ such that if
${\cal I}({\bold x},{\bbeta}_1,{\bbeta}_{2})>M$ for some ${\bold x}\in{\cal X}$,
then there will exists $\widetilde{{\bbeta}}_{2}$ such that
\[
{\cal I}({\bold x},{\bbeta}_1,{\bbeta}_{2})\geq
{\cal I}({\bold x},{\bbeta}_1,\widetilde{{\bbeta}}_{2}), \quad
\forall {\bold x}\in{\cal X},
\]
and
\[
\sup_{{\bold x}\in{\cal X}} {\cal I}({\bold x},{\bbeta}_1,\widetilde{{\bbeta}}_{2})
\leq M.
\]
 \end{assumption}

\begin{theorem}\label{theo:cont}
Assume \ref{ass:contdens2} and \ref{ass:contdens3}.
The KL-criterion (\ref{kl1}) is a locally Lipschitz function and hence it is a
continuous function of $\xi$.
\end{theorem}
\begin{proof}
Let ${\bbeta}_1$ be fixed and $M$ be as in
Assumption~\ref{ass:contdens3}. Define
\begin{equation}
\Theta_2({\bbeta}_1)= \{ \widetilde{{\bbeta}}_{2} \in
\Theta_2 \colon
\sup_{{\bold x}\in{\cal X}} {\cal I}({\bold x},{\bbeta}_1,\widetilde{{\bbeta}}_{2})
\leq M
\}.
\label{theta2_beta1}
\end{equation}
The KL-criterion (\ref{kl1}), for any $\xi\in {\spaceX}$, may be rewritten as
\begin{equation}
I_{2,1}(\xi)=
\inf_{{\bbeta}_{2}\in
\Theta_2({\bbeta}_1)}
\int_{\cal X}
{\cal I}({\bold x},\bbeta_1,\bbeta_2)
\,\xi(d{\bold x}),
\label{kl2}
\end{equation}
for Assumption~\ref{ass:contdens3}.

Let ${\mathcal V}$ be the real vector space of all signed finite measures on
${\cal X} $ (equipped with the usual Borel $\sigma$-algebra ${\cal B}$), which contains
${\spaceX} $ as proper, closed, convex subset. Denoting by $\|h\|_L$ the operator norm of the Banach space of Lipschitz function on ${\cal X}$, the vector space ${\mathcal V}$, equipped with the norm
\[
\|\xi\|_{{\mathcal V}}
= \sup\Big\{ \Big|\int_{\cal X} h({\bold x}) \,\xi(d{\bold x}) \Big|, \|h\|_L\leq 1\Big\},
\]
is a Banach space. (This is a consequence of the results in \cite{Varadarajan} and \cite{Dudley}).


The map $\xi \mapsto
\int_{\cal X} {\cal I}({\bold x},{\bbeta}_1,{\bbeta}_{2}) \,\xi(d{\bold x}),
$ is a linear functional on
${\mathcal V}$. 
Moreover, assumption~\ref{ass:contdens2}
guarantees its boundness on the unit ball $\|\xi\|_{{\mathcal V}} \leq 1 $:
\begin{align*}
\sup_{\|\xi\|_{{\mathcal V}} \leq 1 }
\Big|\int_{\cal X} {\cal I}({\bold x},{\bbeta}_1,{\bbeta}_{2}) \,\xi(d{\bold x})\Big|
& \leq
\sup_{
\substack{
            \|\xi\|_{{\mathcal V}} \leq 1\\
            \|h\|_L \leq \|{\cal I}(\cdot,{\bbeta}_1,{\bbeta}_{2}) \|_L
		}
 }
\Big|\int_{\cal X} h({\bold x}) \,\xi(d{\bold x})\Big|
\\
& =
\|{\cal I}(\cdot,{\bbeta}_1,{\bbeta}_{2}) \|_L
\sup_{
\substack{
            \|\xi\|_{{\mathcal V}} \leq 1\\
            \|g\|_L \leq 1
		}
 }
\Big|\int_{\cal X} g({\bold x}) \,\xi(d{\bold x})\Big|
\\
&=
\|{\cal I}(\cdot,{\bbeta}_1,{\bbeta}_{2}) \|_L ,
\end{align*}
and hence $\xi \mapsto
\int_{\cal X} {\cal I}({\bold x},{\bbeta}_1,{\bbeta}_{2}) \,\xi(d{\bold x})
$ is a continuous function 
(see, e.g., \cite{Borwein}).
The function
\[
\xi\mapsto \inf_{{\bbeta}_{2}\in \Theta_2({\bbeta}_1)}
\int_{\cal X} {\cal I}({\bold x},{\bbeta}_1,{\bbeta}_{2}) \,\xi(d{\bold x})
\]
is concave and upper semi-continuous function since it is the point-wise infimum
 of linear continuous functions.
Moreover, \( -\infty <
\inf_{{\bbeta}_{2}\in \Theta_2({\bbeta}_1)}
\int_{\cal X} {\cal I}({\bold x},{\bbeta}_1,{\bbeta}_{2}) \,\xi(d{\bold x}) <\infty
\)
since, for any $\bbeta_2\in \Theta_2({\bbeta}_1)$,
\[
\Big| \int_{\cal X} {\cal I}({\bold x},{\bbeta}_1,{\bbeta}_{2}) \,\xi(d{\bold x}) \Big| \leq M \,\|\xi\|_{{\mathcal V}} ,
\]
as a consequence of \eqref{theta2_beta1}. Therefore
(see \cite{Borwein}) the function
\[
\xi\mapsto \inf_{{\bbeta}_{2}\in \Theta_2({\bbeta}_1)}
\int_{\cal X} {\cal I}({\bold x},{\bbeta}_1,{\bbeta}_{2}) \,\xi(d{\bold x})
\]
is locally Liptschtz and hence continuous
on the Banach space $({\mathcal V},
\|\cdot\|_{{\mathcal V}} )$.

Recall that the set ${\spaceX}$ is the set of the possible experimental designs $\xi$.
When $\xi_1,\xi_2\in {\spaceX}$, the Kantorovich-Wasserstein  distance $d_w$ can be rewritten also as (see \cite{Gibbs02})
\[
d_w(\xi_1,\xi_2)
= \sup\Big\{ \Big|\int_{\cal X} h({\bold x}) \,d\xi_1({\bold x}) -
\int_{\cal X} h({\bold x}) \,d\xi_2({\bold x}) \Big|, \|h\|_L\leq 1\Big\} = \|\xi_1-\xi_2\|_{{\mathcal V}}.
\]
Hence the KL-criterion \eqref{kl1} is a locally Liptschtz and continuous function on ${\spaceX}$.
\end{proof}

\section{First order algorithm to obtain KL-optimum designs}\label{sect:alg}
To construct a KL-optimum design $\xi^*$, \cite{Tomm2007} propose the use of the first order algorithm:
 \begin{enumerate}
\item[a)]\label{stepa}
Given $\xi_n$, find
 \begin{enumerate}
 \item[a1)]\label{stepa1}
$$
\bbeta_{2,n}=\arg \min_{\bbeta_2 \in \Omega_2}
\int_{\cal X} \int_{\cal Y}
\log
\frac{f_1(y|{\bold x};\bbeta_1)}{f_{2}(y|{\bold x};\bbeta_{2})}\; f_1(y|{\bold x};\bbeta_1) \,dy\, \xi_{n}(d{\bold x})
$$
 \item[a2)]\label{stepa2}
$$
{\bold x}_n=\arg \max_{{\bold x} \in{\cal X} }
\int_{\cal Y} \log \frac{f_1(y|{\bold x};\bbeta_1)}{f_{2}(y|{\bold x};\bbeta_{2,n})}\; f_1(y|{\bold x};\bbeta_1) \,dy
$$
\end{enumerate}
\item[b)]\label{stepb}
Properly choose $0\leq \alpha_n \leq 1$ and construct $\xi_{n+1}=(1-\alpha_n) \xi_{n}+\alpha_n \delta_{{\bold x}_n}$.
\end{enumerate}
Moreover,  \cite{Tomm2007} prove that, for any $\xi$,
$U(\xi)\leq  I_{2,1}(\xi)/I_{2,1}(\xi^*) \leq 1$ where
\begin{equation}
U(\xi)=\left[1+\frac{\max_{{\bold x}\in {\cal X}}   \psi({\bold x};\xi)}{I_{2,1}(\xi)} \right]^{-1},
\label{upper-bound}
\end{equation}
and
\begin{equation}
\psi({\bold x};\xi)=  {\cal I}({\bold x},\bbeta_1, \hat{\bbeta}_2)  -\int_{\cal X}{\cal I}({\bold s},\bbeta_1, \hat{\bbeta}_2)\,\xi(d{\bold s})
 \label{derivata}
 \end{equation}
is the directional derivative of $I_{2,1}(\xi)$ at a regular design $\xi$ in the direction of $\delta_{\bold x}-\xi$ and $\hat{\bbeta}_2$ is the unique element of (\ref{eq:regular}).
According to  \cite{Tomm2007}, the iterative procedure should stop at step $N$ if the upper bound of the efficiency $U(\xi_N)>\delta$,
where $0<\delta<1$ is a  value chosen by the experimenter, for instance $\delta=0.99$.

However,  \cite{Tomm2007} give no specific rule to choose $\alpha_n$, neither they provide any proof of the convergence of the algorithm.
Differently, authors \cite{AMT_Moda10}  suggest to use
\begin{equation}
\alpha_n=\arg\max_{\alpha\in [0,1]} {I_{2,1}} [(1-\alpha) \xi_n + \alpha \delta_{\bold x}]
\label{alpha}
\end{equation}
and using the theory of
point-to-set mappings they provide useful convergence results. As a consequence of Theorem 1 in \cite{AMT_Moda10} and of the continuity of $I_{2,1}(\xi)$ proved in Section \ref{sect:invcont}, the following theorem states that if the algorithm explores regular designs (see \eqref{eq:regular} in Section \ref{sect:frame}), then it converges to the set of KL-optimum designs.

\begin{theorem}\label{theo:conv2}
Let $\xi_0$ be an initial design. For any $n\ge0$, let $\xi_{n+1}$ be one of the designs obtained by the first order algorithm at step $(n+1)$, with $\alpha_n$ as in (\ref{alpha}).

If $\xi_{n}$ is a sequence of KL-regular designs then, as $n\rightarrow \infty$,
$$\inf_{\xi^*\in \arg\max I_{2,1}(\xi)}d_w(\xi_n,\xi^*)\rightarrow 0$$
\noindent and
$$|I_{2,1}(\xi_n) - \max_{\xi}I_{2,1}(\xi)| \rightarrow 0.$$

\noindent In particular, if the set $\{\xi^* \in \Xi: \xi^*= \displaystyle \arg\max_\xi I_{2,1}(\xi)\}$ is a singleton, then $\xi_{n}\to\xi^*$.
\end{theorem}

\begin{proof}
We prove the first part of the theorem by way of contradiction.
Suppose that for every $\epsilon>0$ and for every $N\ge 1$ there exist $n_N\ge N$ such that
\begin{equation}\label{eq:abs}
d_w(\xi_{n_N},\xi^*)\ge \epsilon
\end{equation}
\noindent for all $\xi^*\in \arg\max I_{2,1}(\xi)$.

In \cite[Theorem 1]{AMT_Moda10} it is stated that if $(\xi_n)_n$ is a sequence of regular designs generated by the above described algorithm, then the limit of any converging subsequence of $(\xi_n)_n$ is a KL-optimum design. In \cite[Theorem 1]{AMT_Moda10} it is required that the sub-level $\{ \xi\in{\Xi}\colon I_{2,1}(\xi)\geq I_{2,1}(\xi_0)\}$ is compact. This is always verified as a consequence of Theorem~\ref{theo:cont}: the sub-level is a closed subset of the compact set $\Xi$.

The sequence $(\xi_{n_N})_N$ in \eqref{eq:abs} has at least one convergent subsequence since ${\spaceX}$ is compact; then, from \cite[Theorem 1]{AMT_Moda10} this converges to a KL-optimum design, which contradicts \eqref{eq:abs}.
It follows that for every $\epsilon>0$ there exists $N$ such that for every $n\ge N$
\begin{equation}\label{eq:conv}
\inf_{\xi^*\in \arg\max I_{2,1}(\xi)}d_w(\xi_n,\xi^*)<\epsilon,
\end{equation}
and hence the first part of the theorem is proved. Since $I_{2,1}(\xi)$ is continuous and $\{\xi^* \in \arg\max I_{2,1}(\xi)\}$ is compact, the second part of the result follows from \eqref{eq:conv}.
\end{proof}

A drawback of the application of the first order algorithm is that there is no guarantee that $\xi_{n+1}$ is a regular design in the sense of \eqref{eq:regular}, even if $\xi_n$ is. In the next section we prove that, for generalized linear models (GLM), the algorithm always moves at regular designs; thus it converges, provided that the initial design $\xi_0$ has a non singular information matrix.

\subsection{Regular designs for GLM}\label{sect:GLM}
As recalled in Section \ref{sect:frame}, if the KL-criterion is used to find an optimum design to discriminate between two models, a regular design is such that there is a unique estimate of the parameters  $\bbeta_2$ as in \eqref{eq:regular}.
In many situations, however, the aim of an experiment is to estimate as precisely as possible the parameters (or a function of the parameters) of the assumed known model. In this context a design is called regular if its Fisher information matrix is non singular.
In this section it is proved that these two different definitions of \lq\lq regular design'' are equivalent
for a large class of useful statistical models.

In what follows, we specialize the notation for a GLM.
${\bold X}$ is the design matrix whose $i$-th row is ${\bold x}_{i.}^T$ (i.e. ${\bold x}_{i.}$ is the $d_2 \times 1$ vector of experimental conditions for the $i$-th unit),
$\bbeta_2\in\Theta_2$ is the $d_2 \times 1$ vector of regression parameters of $f_2(y|{\bold x},\bbeta_2)$,
$\mu_i={\rm E}_{2} (Y|{\bold x}_{i.})$ and ${\rm Var}_i(Y)={\rm Var}_{2}(Y|{\bold x}_{i.})$ are the response mean and the response variance under model $f_2(y|{\bold x}_{i.},\bbeta_2)$, and
$\eta_i={\bold x}_{i.}^T \bbeta_2$ is the linear predictor.
Both $\mu_i$ and ${\rm Var}_i(Y)$ depend on ${\bold x}_{i.}$ and $\bbeta_2$ through $\eta_i$, in particular $g(\mu_i)=\eta_i$ where  $g(\cdot)$ is a link function.
 It is well known that for the generalized linear model the Fisher information matrix is
\begin{equation}
{\bold J}={\bold X}^T {\bold W} {\bold X},
\label{inf-matrix}
\end{equation}
where
$$
{\bold W}={\rm diag}\left[ \frac{1}{{\rm Var}_i(Y)}  \left(  \frac{\partial \mu_i}{\partial \eta_i}  \right)^2 \right]
$$
is a diagonal matrix which depends on both the design matrix ${\bold X}$ and the parameter vector $\bbeta_2$.

From (\ref{inf-matrix}), ${\bold J}=\tilde {\bold X}^T \tilde {\bold X}$, where $\tilde {\bold X}={\bold W}^{1/2} {\bold X}$; thus
$${\rm Range}\; {\bold J}= {\rm Range}\; \tilde {\bold X}^T= {\rm Range}\; {\bold X}^T,$$
since ${\bold W}$ is a diagonal non-singular matrix. Therefore, the Fisher information matrix is non-singular (i.e. ${\rm Rank}({\bold J})=d_2$) if and only if
 ${\rm Rank}({\bold X})=d_2$.

As recalled in Section \ref{sect:frame}, in the context of discrimination between rival models, 
an exact design $\xi=\frac{1}{n}\sum_i \delta_{{\bold x}_{i.}}$ 
is regular if the set
$$
\Omega_2 \Big(\frac{1}{n}\sum_i \delta_{{\bold x}_{i.}} \Big) =
\Big\{
\widehat{\bbeta}_2:  \widehat{\bbeta}_2  =  \arg\min_{\bbeta_2} \sum_{i=1}^n \int_{\cal Y} \log \frac{f_1(y|{\bold x}_{i.},\bbeta_1)}{f_2(y|{\bold x}_{i.},\bbeta_2)} f_1(y|{\bold x}_{i.},\bbeta_1)\, dy
\Big\}
$$
is singleton. Since
\begin{eqnarray}
& &  \sum_{i=1}^n \int_{\cal Y} \log \frac{f_1(y|{\bold x}_{i.},\bbeta_1)}{f_2(y|{\bold x}_{i.},\bbeta_2)} f_1(y|{\bold x}_{i.},\bbeta_1)\, dy \nonumber\\
&=&
 \sum_{i=1}^n \int_{\cal Y} \log f_1(y|{\bold x}_{i.},\bbeta_1)\,  f_1(y|{\bold x}_{i.},\bbeta_1)\, dy  -
 \sum_{i=1}^n \int_{\cal Y} \log f_2(y|{\bold x}_{i.},\bbeta_2)\,  f_1(y|{\bold x}_{i.},\bbeta_1)\, dy, \nonumber
\end{eqnarray}
finding
$$
\arg\min_{\bbeta_2}  \sum_{i=1}^n \int_{\cal Y} \log \frac{f_1(y|{\bold x}_{i.},\bbeta_1)}{f_2(y|{\bold x}_{i.},\bbeta_2)} f_1(y|{\bold x}_{i.},\bbeta_1)\, dy
$$
is equivalent to compute $\displaystyle\arg\max_{\bbeta_2}\,{\cal ML}_2(\bbeta_2)$ where $\xi=\frac{1}{n}\sum_i \delta_{{\bold x}_{i.}}$
and thus
$$
{\cal ML}_2(\bbeta_2)=  \sum_{i=1}^n \int_{\cal Y} \frac{1}{n}\log f_2(y|{\bold x}_{i.},\bbeta_2)\, f_1(y|{\bold x}_{i.},\bbeta_1)\, dy.
$$
Assuming that ${\cal ML}_2(\bbeta_2)$ is differentiable with respect to $\bbeta_2$ for any fixed ${\bold x}_{i.}$,
then the maximization of ${\cal ML}_2(\bbeta_2)$ is
performed by setting the partial derivatives of ${\cal ML}_2(\bbeta_2)$ equal to zero.
Let us assume that
\begin{eqnarray}
 \frac{\partial {\cal ML}_2(\bbeta_2)}{\partial \beta_{2j}}
& \propto &
\sum_{i=1}^n \frac{\partial }{\partial \beta_{2j}} \int_{\cal Y}   \log f_2(y|{\bold x}_{i.},\bbeta_2)\, f_1(y|{\bold x}_{i.},\bbeta_1)\, dy \nonumber \\
& = &
 \sum_{i=1}^n \int_{\cal Y}   \frac{\partial }{\partial \beta_{2j}} \log f_2(y|{\bold x}_{i.},\bbeta_2)\, f_1(y|{\bold x}_{i.},\bbeta_1)\, dy .
\label{mean-log-likelihood}
\end{eqnarray}
It is well known that if $f_2(y|{\bold x},\bbeta_2)$ is a GLM, then at ${\bold x}_{i.}$, the $i$-th experimental condition,
$$
\frac{\partial \log f_2(y|{\bold x}_{i.},\bbeta_2) }{\partial \beta_{2j}}= \frac{(y-\mu_i) {x}_{ij}}{{\rm Var}_{i}(Y)} \left(\frac{\partial \mu_i}{\partial \eta_i}\right), \quad j=1,\ldots,d_2.
$$
Therefore, from equation (\ref{mean-log-likelihood})
\begin{eqnarray}
\frac{\partial {\cal ML}_2(\bbeta_2)}{\partial \beta_{2j}}
&\propto &
  \sum_{i=1}^n \int_{\cal Y} \frac{(y-\mu_i) x_{ij}}{{\rm Var}_i(Y)} \left(\frac{\partial \mu_i}{\partial \eta_i}\right) \, f_1(y|{\bold x}_{i.},\bbeta_1)\, dy \nonumber \\
&=&
  \sum_{i=1}^n  \frac{[{\rm E}_{1}(Y|{\bold x}_{i.})-\mu_i] x_{ij}}{{\rm Var}_i(Y)} \left(\frac{\partial \mu_i}{\partial \eta_i}\right)\nonumber \\
&=&
{\bold v}^T({\bold {\bold X}};\bbeta_2) \,{\bold x}_{.j}, \quad j=1\ldots,d_2,
\label{system}
\end{eqnarray}
where ${\bold x}_{.j}$ denotes the $j$-th column of the design matrix $ {\bold X}$ and ${\bold v}({\bold X};\bbeta_2)$ is a $n\times 1$ vector 
whose $i$-th item is $$v_i({\bold x}_{i.};\bbeta_2)=\displaystyle \frac{{\rm E}_{1}(Y|{\bold x}_{i.})-{\rm E}_{2}(Y|{\bold x}_{i.}) }{{\rm Var}_i(Y)} \left(\frac{\partial \mu_i}{\partial \eta_i}\right)$$
and ${\rm E}_{1}(Y|{\bold x}_{i.})$ is the response mean under  the known model $f_1(y|{\bold x}_{i.},\bbeta_1)$.


Setting equal to zero the right-hand side of equations (\ref{system}), we have  a system of $d_2$ non-linear equations whose solution is the maximum likelihood estimator.
For the most commonly used GLMs (for instance when the link function is canonical and for the link functions described in \cite{Wedderburn1976}) this solution exists and is unique if and only if ${\rm Rank}({\bold X})=d_2$ i.e. if and only if the Fisher information matrix is non-singular.



Let us denote by \lq\lq regular GLM'' a GLM for which there exists a unique solution of the system  ${\bold v}^T({\bold X};\bbeta_2) \,{\bold x}_{.j}=0$, $j=1,\ldots,d_2$ if and only if ${\rm Rank}({\bold X})=d_2$. Theorem \ref{theo:conv2} has the following corollary at least for  regular GLMs.

\begin{corollary} Let $f_1(y|{\bold x},\bbeta_1)$ and $f_2(y|{\bold x},\bbeta_2)$ be regular GLMs and let $\xi_0\in  {{\spaceX}} $ have a non-singular Fisher information matrix.
For any $n\ge0$, let $\xi_{n+1}$ be one of the designs obtained by the first order algorithm at  step $(n+1)$, with $\alpha_n$ as in (\ref{alpha}).
Then
$$|I_{2,1}(\xi_n) - \max_{\xi}I_{2,1}(\xi)| \rightarrow 0$$
\noindent as $n\rightarrow \infty$. In particular, if the optimum $\xi^*$ is unique, $\xi_{n}\to\xi^*$.
\end{corollary}
\begin{proof}
In the classes of the GLMs, any KL-regular design has non-singular Fisher information matrix, and any design with non-singular Fisher information matrix is KL-regular;
therefore if $\xi_{n}$ is regular then $\xi_{n+1}$ is also regular. For these models, it is then guaranteed that the algorithm moves on regular designs if it starts from a regular design $\xi_0$.
\end{proof}

\section{Examples and computational aspects} \label{sect:ex}
In this section we show some practical implications of the theoretical results proved in Sections \ref{sect:invcont} and \ref{sect:alg}.

For what concerns the invariance property, the practitioner can observe that the algorithm performs differently in different experimental regions. 
If it is numerically difficult to compute the KL-optimum design in a specified ${\cal X}$, the experimenter
can take advantage of the invariance property.
He can compute the KL-optimum design in a suitable transformed space ${\cal Z}$ (where less computational effort is required) and then find the KL-optimum design in ${\cal X}$ through the invariance property.
In the example provided in Subsection \ref{sub:bench} the algorithm needs more iterations than in a transformed space.

Continuity of KL-criterion is necessary to assure that the approximated solution obtained with the stopping rule is close to the KL-optimum design. Furthermore, it also guarantees that the solution obtained through the regularization procedure is a ``nearly'' KL-optimum design. Subsection \ref{ex:6.2} gives an example where the problem is regularized since the KL-optimum design is singular.

In Subsection \ref{sub:comp} some practical hints are proposed, in comparison with the suggestions provided in \cite{Fedorov1972}.
\subsection{A benchmark test}\label{sub:bench}
In this example we develop a benchmark test for the convergence of the algorithm, where the analytical result for the KL-optimum design is obtained through the Chebyshev polynomials.

We take ${x}\in {\cal X}=[-1;1]$ and we test a polynomial of degree three versus a polynomial of degree two. In other words,  $f_1(y|{\bold x},\bbeta_1)$ and $f_2(y|{\bold x},\bbeta_2)$ are the probability densities of two Gaussian random variables with means $\beta_{01}+\beta_{11}{ x}+\beta_{21}{ x}^2+\beta_{31}{ x}^3$ and
$\beta_{02}+\beta_{12}{x}+\beta_{22}{ x}^2$, respectively and common variance $\sigma^2$. We assume the polynomial of degree 3 is completly known. Except for the constant of proportionality $\beta_{31}$, given a design $\xi$, the Kullback-Leibler divergence between these two Gaussian models is
${\cal I}({x},\btheta)=\left(\theta_0+\theta_1 { x}+\theta_2 { x}^2+{ x}^3 \right)^2 $,
where $\theta_i=\frac{\beta_{i1}-\beta_{i2}}{\beta_{31}}$ and $i=0,1,2$.
Let \[
{\cal J}(\xi,{\btheta})=\int_{\cal X} \left(\theta_0+\theta_1 {x}+\theta_2 { x}^2+{ x}^3 \right)^2 \xi(dx),
\] then, except for a constant of proportionality, the KL-criterion is
\[I_{2,1}(\xi)=\inf_{{\btheta}} {\cal J}(\xi,{\btheta}).
\]
In  order to solve the previous optimization problem we set the partial derivatives of ${\cal J}(\xi,\btheta)$ equal to zero, i.e.
\begin{equation}
\frac{\partial{\cal J}(\xi,{\btheta})}{\partial\theta_i}=2 \int_{\cal X} {x}^i \left( \theta_0+\theta_1 { x}+\theta_2 { x}^2+{x}^3  \right) \xi(dx)=0,\quad i=0,1,2.
\label{hat-theta}
\end{equation}
The Chebyshev polynomial of degree three, $T_3({ x})=4x^3-3x$, fulfills the following orthogonality conditions,
\[
 \int_{\cal X} { x}^i\, T_3({x}) \,\tilde\xi(dx)=0,\quad i=0,1,2,
\]
where
$\tilde{\xi}=\left\{ \begin{array}{cccc}
-1 & -1/2 & 1/2 & 1 \\
1/6 & 1/3 & 1/3 & 1/6
\end{array}  \right\}$. Let us note that the design points of $\tilde{\xi}$ are the singular points of $T_3({ x})$ in ${\cal X}=[-1,1]$.

Therefore, if we denote  the solutions of Equation (\ref{hat-theta}) by $\hat{\btheta}=(\hat\theta_0,\hat\theta_1,\hat\theta_2)^T$,
we have that $\hat\theta_0=\hat\theta_2=0$ and $\hat\theta_1=-3/4$.

Let us prove that $\tilde{\xi}$  is actually a KL-optimum design by checking the equivalence theorem inequality (see \cite{Tomm2007}, for more details).
Since $T_3({x})$ is a Shabat polynomial (i.e. it assumes the same minimum and maximum values except for the sign) we have that the directional derivative of $I_{2,1}(\xi)$ at $\tilde{\xi}$ in the direction of $\xi_x-\tilde\xi$ is
$$
\psi({ x};\tilde{\xi})= {\cal I}({ x},\hat{\btheta})
 -\int_{\cal X}{\cal I}({x},\hat{\btheta})\,\tilde\xi(dx)  =\left( -\frac{3}{4}{ x} +{ x}^3 \right)^2-\max_x \left( -\frac{3}{4}{ x} +{ x}^3 \right)^2\leq 0
$$
with equality at the singular points -1, -1/2, 1/2 and 1. This proves that $\tilde{\xi}$ is a KL-optimum design.
The same result has been shown in a different way by \cite{Spruill} and extended by \cite{Dette}.

 We have applied the first order algorithm described in Section \ref{sect:alg} using Matlab 7.6 optimization toolbox, obtaining the following results. The iterative procedure converges to the KL-optimum design $\tilde{\xi}$ after $107$ iterations with $\delta = 0.99$, and after $384$ iterations
with $\delta=0.995$.

We have then tested the invariance property in this benchmark setting.
If the experimental condition is $z=2+4 { x} \in{\cal Z}=[-2;6]$ instead of ${x}\in{\cal X}$, then 
after $430$ iteration with $\delta=0.95$ the KL-optimum design
${\eta}^*=\left\{ \begin{array}{cccc}
-2 & 0 & 4 & 6 \\
1/6 & 1/3 & 1/3 & 1/6
\end{array}  \right\}$ is reached. This solution 
is consistent with the invariance property given in Theorem \ref{theo:inv}.

Note that in the transformed experimental region the optimum is reached with a larger number of iterations.

\subsection{Discrimination between two logistic regression models}
\label{ex:6.2}

The purpose of this Subsection is to handle an example where the optimal design is singular.
The experimental conditions are assumed to vary in
 the interval ${\cal X}=[0,1]$ and logistic regression models are considered
 instead of linear regression models. In other words,
  $Y$ is  a binary response variable such that
 $$P(Y=1|{ x};\bbeta_i)=F(\eta_i)=\frac{e^{\eta_i}}{1+e^{\eta_i}},\quad i=1,2$$ where
$\eta_1 = \beta_{01}+\beta_{11} {x}+\beta_{21} { x}^2$ and $\eta_2 = \beta_{12} { x}+\beta_{22} { x}^2$
are two rival models for the expected response.
For these models the Kullback-Leibler divergence is
\[
{\cal I}({ x},\bbeta_1,\bbeta_2)=(\eta_1-\eta_2)\frac{\exp(\eta_1)}{1+\exp(\eta_1)}+\log \frac{1+\exp(\eta_2)}{1+\exp(\eta_1)}.
\]
The algorithm reaches in few steps the one point design which concentrates the whole mass at zero, i.e. $\xi^*=\delta_0$. Once the algorithm has reached this design, it fails in capturing the whole $\Omega_2(\xi^\ast)$, and it enters in a loop.

%
It is well known that a design $\xi^*$ is optimum if and only if  $\psi({\bold x},\xi^*)\leq 0$ for any $ {\bold x}$. However, $\psi({\bold x},\xi)$  can be computed only for regular designs.
When the optimum design is singular, in order to check its optimality we can regularize the problem (see \cite{AMT_Moda10}), that is to use the function
\begin{equation}\label{eq:reg1}
I_\gamma(\xi)=I_{2,1}[(1-\gamma)\xi+\gamma \tilde\xi]
\end{equation}
instead of $I_{2,1}(\xi)$, where $0<\gamma<1$ and  $\tilde\xi$ is a regular design, that is a design with a non singular information matrix.
Let  $\xi_1=(1-\gamma)\xi+\gamma \tilde\xi$; then,
$I_\gamma(\xi)=I_{21}(\xi_1)$.
The first order algorithm may be specialized for $I_\gamma(\xi)$ using $\xi_{1n}=(1-\gamma)\xi_n+\gamma \tilde\xi$ instead of $\xi_n$ at the step (a1) of the algorithm described in Section \ref{sect:alg}.
The stopping rule recalled in Section \ref{sect:alg} may also be specialized for $I_\gamma(\xi)$. \\
Let  $\bar\xi_1=(1-\gamma)\bar\xi+\gamma\tilde\xi$; then, the
directional derivative of $I_\gamma(\xi)$  in the direction of $\bar\xi-\xi$ is
\begin{equation}
\partial I_\gamma(\xi;\bar\xi)=\partial I_{21}(\xi_1;\bar\xi_1)= \int_{\cal X}  \!  \psi({\bold x};\xi_1) \bar\xi_1(d{\bold x}),
\label{der-dir1}
\end{equation}
where the last equality is proved by \cite{Tomm2007}.
From  (\ref{der-dir1}) and (\ref{derivata}), it follows that
\begin{eqnarray}\label{eq:reg3}
\partial I_\gamma(\xi;\bar\xi)&=&\!
\int_{\cal X}  \! \left[
 {\cal I}(f_1,f_2,{\bold x},\hat{\bbeta}_2)  -\int_{\cal X}{\cal I}(f_1,f_2,{\bold s},\hat{\bbeta}_2)\,\xi_1(d{\bold s}) \right]
\bar\xi_1(d{\bold x})\nonumber\\
&=&
\!\int  (1-\gamma) \psi({ {\bold x}};\xi)  \bar\xi(d{\bold x}),
\label{der-dir2}
\end{eqnarray}
where $\psi({{\bold x}};\xi)= \left[  {\cal I}(f_1,f_2,{\bold x},\hat{\bbeta}_2)  -\int_{\cal X}{\cal I}(f_1,f_2,{\bold s},\hat{\bbeta}_2)\,\xi(d{\bold s})   \right]$ with
\begin{equation}
\hat \bbeta_2= \bbeta_2(\xi_1)=\arg
\min_{\bbeta_2 \in \Omega_2} \int_{\cal X} \,{\cal I}(f_1,f_2,{\bold x},\bbeta_2) \, \xi_1(d{\bold x}).
\label{teta2}
\end{equation}
instead of $\hat \bbeta_2= \bbeta_2(\xi)$.

Since $I_\gamma(\xi)$ is a concave function, the iterative procedure based on $I_\gamma(\xi)$   stops at the step $n$ if
\begin{equation}
\left[1+\frac{\max_{{\bold x}\in {\cal X}} (1-\gamma)  \psi({\bold x};\xi_n)}{I_{\gamma}(\xi_n)} \right]^{-1}>\delta,
\label{eq8}
\end{equation}
where $0<\delta<1$ is a suitable value chosen by the experimenter, e.g. $\delta=0.995$.\\[.3cm]
It is well known that in maxmin problems it may be difficult to check the optimality of a design because the
expression of the directional derivative is not easy to be found when a design is not regular (see, for instance,  \cite[Theorem~2.6.1]{FedorovHackl}). More specifically, in order to check the directional derivative it is necessary introduce a measure on $\bbeta_2$, which is not easy to be found. Differently, since $I_\gamma(\xi)=I_{21}(\xi_1)$ and
$\xi_1$ is a  regular design by definition, equation (\ref{der-dir2}) provides the directional derivative of
$I_\gamma(\xi)$ at any design $\xi$, without assuming the regularity of $\xi$.
Therefore, if the above described algorithm stops at the $n$-th step, it is always possible to check the optimality of $\xi_n$ through the following directional derivative of $I_\gamma(\xi)$ at $\xi_n$ in the direction of $\delta_{\bold x}-\xi_n$,
\begin{equation}
\psi_\gamma({\bold x};\xi_n)=(1-\gamma)\psi({\bold x};\xi_n).
\label{der-dir2bis}
\end{equation}
If   $\psi_\gamma({\bold x};\xi_n)\leq 0$ for any ${\bold x}$ and the equality is reached at the design points of $\xi_n$, then $\xi_n$ is an optimal design.

Applying the first order algorithm to $I_\gamma (\xi)$, after only two iterations we have that $\xi_{n}^*=\delta_0$, as shown by Figure \ref{fig-der} which displays the function
$\psi_\gamma({x};\xi_n^*)$.
\begin{figure}[!!ht]
\begin{center}
\includegraphics[width=8cm,height=4.5cm]{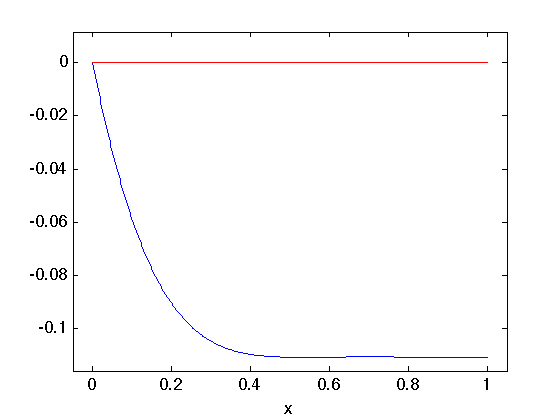}
\caption{Directional derivative $\psi_\gamma({x};\delta_0)$}
\end{center}
\label{fig-der}
\end{figure}
As proved in \cite{AMT_Moda10}, from the continuity of the KL-criterion $\xi_{n}^*$ is a ``nearly''
KL-optimal design  and, in particular, the smaller $\gamma$  the better this design. Note that different regular designs $\tilde \xi$ lead to different expressions of $\psi_\gamma({\bold x};\xi_n^*)$ but to the same ``nearly'' KL-optimum design $\xi_n^*=\delta_0$.

\subsection{Computational aspects}\label{sub:comp}
In this section we suggest some tricks to get numerically a KL-optimum design through the first order algorithm.
We also compare these tricks to \cite[pag.~109ff.]{Fedorov1972}, underlying possible improvements.
\begin{enumerate}
\item
When $\xi_n$ is singular, the directional derivative of the KL-criterion function at $\xi_n$
in the direction of $\delta_{\bold x}-\xi_n$
 cannot be expressed as in \eqref{derivata} with $\xi=\xi_n$ (for more details see \cite[pag.~41]{FedorovHackl}).
 Therefore, it cannot be used in the algorithm.
On the other hand, the directional derivative of the regularized criterion \eqref{eq:reg1} at $\xi_n$,
in the direction of $\delta_{\bold x}-\xi_n$, is given by \eqref{der-dir2bis}, whether or not $\xi_n$ is regular.
Therefore, before each step (a1), we suggest to decide between the original KL-criterion or the regularized one according to the following rule:
use the regularized criterion if $\xi_n$ has less than $d_2$ support points (see the discussion in Section~\ref{sect:GLM}); otherwise, use the original KL-criterion.
\item
At step {(a1)} of the first order algorithm described in Section \ref{sect:alg}, in order to find $\bbeta_{2,n}$ a starting point  $\bbeta_{2,n}^{(0)}$ is necessary.
We suggest to use a point close to $\bbeta_{2,n-1}$, for instance $\bbeta_{2,n}^{(0)}=\bbeta_{2,n-1}+{\boldmath \varepsilon}$, where ${\boldmath\varepsilon}$ is a random error with the magnitude of $\bbeta_{2,n-1}$.
Sometimes, the reasonable choice of $\bbeta_{2,n}^{(0)}=\bbeta_{2,n-1}$ led the algorithm to a situation of stagnation around a local maximum, for this reason a random error is also considered.
\item
As in \cite{Fedorov1972}, at step {(a2)}, in order to choose a starting point ${\bold x}_n^{(0)}$ we suggest to take the maximum point in a regular grid or in a set of randomly selected starting points of  ${\cal X}$, as ${\cal X}$ is a compact space. In our benchmarks, we did not find significant differences between these two choices.
\item
We suggest to check the terminal condition $U(\xi)>\delta$ after step {(a2)}, where the point ${\bold x}_n$ is computed, which gives the direction of the maximum increasing $\delta_{{\bold x}_n}-\xi_{n-1}$.
Note that if  $\xi_n$ is a singular design, even though it has at least $d_2$ support points, then it is not computationally guaranteed that $\delta_{{\bold x}_n}-\xi_{n-1}$ is an increasing direction. In that case, the value $\alpha_n$ computed at step {(b)} is zero, and thus we suggest to regularize the criterion function, as described in Example \ref{ex:6.2}.
\item
At step {(b)} we construct a new design $\xi_{n+1}$ adding the new point ${\bold x}_n$ with weight $\alpha_n$.
Before going back to step {(a1)}, we suggest
\begin{enumerate}
\item[(i)]
as in \cite{Fedorov1972}, to collapse the points ${\bold x}_s\in B({\bold x}_n,r_n)$ (where $r_n$ is a radius tending to zero) in a new point $\bar {\bold x}_n$ with weight given by the sum of the weights of ${\bold x}_s$'s and ${\bold x}_n$.
To speed up the algorithm we suggest to increase the  weight of ${\bold x}_n$ in the computation of the barycenter
$\bar {\bold x}_n$, as $n$ becomes larger (we found that a good choice is to increase the weight of ${\bold x}_n$ by a factor of magnitude $n^{0.8}$, while the radius decreases as $n^{-0.65}$);
\item[(ii)]
as in \cite{Fedorov1972}, to remove the support points with low weights (in absolute value as in \cite{Fedorov1972}, and also when the weight is low compared with the mean weight of the others).
\end{enumerate}
\end{enumerate}

\section{Conclusions}\label{sect:end}

In this paper we have provided some nice properties for the KL-optimality criterion, which are useful from a practical point of view.
For instance, since the KL-optimum design is invariant to a scale-position transformation of the design region, if the experimental domain changes from ${\cal X}$ to ${\cal Z}=\{{\bold z}: {\bold z}={\bold a}+{\bold B}{\bold x}|{\bold x}\in{\cal X}\}$, in order to compute a new KL-optimum design on ${\cal Z}$ it is enough to  change in the same way the support points of the KL-optimum design on ${\cal X}$.

We have proved a crucial theoretical property, that is, the KL-criterion is continuous with respect to the design $\xi$. Continuity also guarantees the computational stability, since when designs are obtained in practice they are approximated at each step.  Moreover it is a key property to prove the convergence of the first order algorithm.

We have showed that, at least in a large class of models, if the first order algorithm starts from an initial design with non-singular information matrix, then it moves on regular designs and thus it converges to the set of KL-optimum designs.
If the KL-optimum design is singular then, in order to check its optimality, a regularized KL-criterion can be used: 
from the continuity of the KL-criterion, an optimum design for the regularized  criterion is also nearly optimum for the KL-criterion.

Let us note that Kullback (see \cite{Kullback}) has already investigated some invariance properties of the Kullback-Leibler divergence and some relations with the Fisher information matrix. The KL-criterion, however, is the minimum Kullback-Leibler divergence between two parametric families. Therefore, the invariance property and the connection with the Fisher information matrix herein proved, are original results.

In this work, we also provide a list of useful advices to implement the first order algorithm, in order to get numerically a KL-optimum design. Some interesting examples show the application of the  algorithm.

The KL-criterion depends on the unknown model parameters. For this reason, KL-optimum designs computed in this work are only locally optimal. In \cite{Dette2013} it is showed that T-optimal designs (which are also KL-optimal) are sensitive with respect to parameter misspecification. To construct efficient discriminating designs which are instead robust to parameter misspecification, \cite{TommasiFidalgo} propose a Bayesian version of the KL-criterion, while \cite{TommasiMay} follow a sequential approach. Another possibility to solve the problem of parameter dependence could be to use the following standardized maxmin KL-optimality criterion:
$$I_{2,1}^M(\xi)=\inf_{\bbeta_1 \in \Theta_1}\dfrac{I_{2,1}(\xi;\bbeta_1)}{I_{2,1}(\xi^*;\bbeta_1)}.$$

Another drawback of the KL-criterion is that the Kullback-Leibler divergence is not simmetric and thus,
when the  rival models are not nested, two different KL-criteria may be defined. This problem is solved in \cite{TommasiFidalgo} by using a prior distribution for the models. Another solution could be to use
the Jeffreys divergence instead of the Kullback-Leibler one, as done in \cite{DragFedorov}.

As a matter of future work we intend to develop a new optimality criterion based on the Jeffreys divergence as well as to apply the standardized maxmin approach to compute robust discriminating design.

\bigskip

\begin{footnotesize}
\noindent\textbf{Acknowledgements.}
We are grateful to Professor Valerii Fedorov who suggested us some relevant references and to two anonymous referees for their useful remarks. We thank also Professor Giancarlo Manzi for helping us in improving the original draft of the paper.
\end{footnotesize}

\end{document}